\documentclass[aps,onecolumn,nofootinbib,tightenlines,superscriptaddress]{revtex4}

\usepackage{amssymb,amstext,amsmath,amsthm,slashed}
\usepackage[dvips]{graphicx}
\usepackage{latexsym}
\usepackage{psfrag}
\usepackage{amsfonts}
\usepackage{color}
\usepackage{verbatim}
\usepackage{bbm}
\usepackage{dcolumn}
\usepackage{tabularx}
\usepackage{leftidx}
\usepackage{tensor}
\usepackage{boxedminipage}
\usepackage{fancyvrb}
\usepackage{float}
\usepackage{wrapfig}
\usepackage{booktabs}
\usepackage{multirow}
\usepackage{subfigure}
\usepackage{dsfont} 
\usepackage{mathrsfs}
\usepackage[hidelinks]{hyperref}
\usepackage{scalerel,stackengine}

\newtheorem{Def}{Definition}[section]

\newtheorem{Prp}[Def]{Proposition}

\newtheorem{Example}[Def]{Example}

\newcommand{\1}{\mbox{\rm 1 \hspace{-1.05 em} 1}}
\newcommand{\Sl}{\mbox{$\prec \!\!$ \nolinebreak}}
\newcommand{\Sr}{\mbox{\nolinebreak $\succ$}}
\newcommand{\fnsep}{\textsuperscript{,}}

\begin{document}

\title{A Quasi-local, Functional Analytic Detection Method for \\ Stationary Limit Surfaces of Black Hole Spacetimes \vspace{0.1cm}}

\author{Christian R\"oken\vspace{0.25cm}} 

\email{croeken@uni-bonn.de}

\affiliation{Lichtenberg Group for History and Philosophy of Physics - Institute of Philosophy, University of Bonn, 53113 Bonn, Germany \vspace{0.25cm}}

\affiliation{Department of Geometry and Topology - Faculty of Science, University of Granada, 18071 Granada, Spain \vspace{0.25cm}}

\date{December 2022 / December 2025}

\begin{abstract}
\vspace{0.4cm} \noindent \textbf{\footnotesize ABSTRACT.} \hspace{0.04cm} We present a quasi-local, functional analytic method to locate and invariantly characterize the stationary limit surfaces of black hole spacetimes with stationary regions. The method is based on ellipticity-hyperbolicity transitions of the Dirac, Klein--Gordon, Maxwell, and Fierz--Pauli Hamiltonians defined on spacelike hypersurfaces of such black hole spacetimes, which occur only at the locations of stationary limit surfaces and can be ascertained from the behaviors of the principal symbols of the Hamiltonians. Therefore, since it relates solely to the effects that stationary limit surfaces have on the time evolutions of the corresponding elementary fermions and bosons, this method is profoundly different from the usual detection procedures that employ either scalar polynomial curvature invariants or Cartan invariants, which, in contrast, make use of the local geometries of the underlying black hole spacetimes. As an application, we determine the locations of the stationary limit surfaces of the Kerr--Newman, Schwarzschild--de Sitter, and Taub--NUT black hole spacetimes. Finally, we show that for black hole spacetimes with static regions, our functional analytic method serves as a quasi-local event horizon detector and gives rise to relational concepts of event horizons and black hole entropy. 
\end{abstract}

\setcounter{tocdepth}{2}

\vspace{0.1cm}

\maketitle

\tableofcontents

\section{Introduction} \label{SectionI}

\noindent In classical general relativity, the canonical key constituents associated with a black hole spacetime are a (nonremovable) curvature singularity at the core of the black hole, an (outer) event horizon that forms the boundary surface of the black hole with the exterior spacetime, and an ergosphere, which is a region where all observers are forced to co-rotate with the black hole, and therefore cannot remain stationary (see, e.g., \cite[Chapter 9.3]{HawkingEllis} and \cite[Chapters 12.1 and 12.3]{Wald}). While the detection of such a curvature singularity is standardly carried out by either identifying the infinities of certain scalar curvature invariants (e.g., \cite{EllisSchmidt, Thorpe}) or by examining the geodesic completeness of a given black hole spacetime (e.g., \cite{Penrose1, Geroch1, Geroch2}), the detection of an event horizon or an ergosurface, i.e., the outermost boundary of an ergosphere, is more subtle. 

This is because the usual notions of a black hole event horizon and of an ergosurface are global, observer-independent concepts that are based on conformal infinity of the black hole spacetime, and are thus teleological in nature (see, e.g., \cite{Booth}). To be more precise, the (outer) black hole event horizon $\mathfrak{H}$ of an asymptotically flat black hole spacetime $(\mathfrak{M}, \boldsymbol{g})$ may be defined as the future boundary of the chronological past of future null infinity $\partial I^-(\mathscr{I}^+) \cap \mathfrak{M}$, that is, the future boundary of the region from which one can reach future null infinity along a future-directed, timelike curve \cite{Penrose2}. (For equivalent notions and for a definition of an event horizon in a black hole spacetime that is not asymptotically flat see \cite{GibbonsHawking}.) The ergosurface $\mathcal{E}$ of an asymptotically flat black hole spacetime, on the other hand, is determined by the set $\bigl\{\boldsymbol{x} \in \mathfrak{M} \, \big| \, \boldsymbol{g}(\boldsymbol{K}, \boldsymbol{K})_{| \boldsymbol{x}} = 0\bigr\}$ complying with the usual condition for stationary limits, where $\boldsymbol{K}$ is a unique time-translational Killing vector field that is timelike near spacelike infinity~$i^0$.

To avoid the problem of localization that comes along with such global concepts, one can instead of, e.g., the former event horizon concept employ the notion of a quasi-local black hole horizon, for instance, a trapping horizon \cite{Hayward}, an isolated horizon \cite{ABDFKLW, AshtekarKrishnan2}, or a dynamical horizon \cite{AshtekarKrishnan, AshtekarKrishnan2}, and the associated quasi-local detection methods. These quasi-local black hole horizons preserve all essential and un\-pro\-blematic features of the global black hole event horizon concepts, but they require neither the knowledge of spacetime as a whole nor the condition of asymptotic flatness. Accordingly, they can be used to quasi-locally characterize and distinguish black holes. (For a comprehensive review on the various notions of quasi-local black hole horizons, we again refer to \cite{Booth}.) The typical methods for the detection of these quasi-local black hole horizons evaluate the zeros of certain scalar polynomial curvature invariants \cite{AbdelqaderLake, Lake, McNuttPage,PageShoom, TavlayanTekin} or of Cartan invariants \cite{BrooksEtAl}, as these invariants are nonzero away from the horizons but vanish on the horizons. (Traversing from the outside to the inside of these surfaces, they change their sign.) The scalar polynomial curvature invariants found in \cite{KarlhedeEtAl, GassEtAl}, on the other hand, vanish on stationary limit surfaces. Thus, aside from the event horizons of black hole spacetimes with static regions, they can also be used to locally detect the ergosurfaces of black hole spacetimes with stationary, rotational regions. The main field of application of the above invariants is to be found in numerical relativity, where they are frequently employed as tools to track quasi-local horizons in the numerical evolution of black holes that have settled down to becoming either stationary or near-stationary, providing the horizons' exact or approximate locations on spacelike slices of the black hole spacetimes, respectively (see, e.g., \cite{PetersColeySchnetter} for an application to black hole mergers). 

In the present paper, we introduce a quasi-local, functional analytic method to detect the stationary limit surfaces of black hole spacetimes with stationary regions.\footnote{A region $\mathcal{U} \subseteq \mathfrak{M}$ of a spacetime $(\mathfrak{M}, \boldsymbol{g})$ is called stationary if it admits a unique time-translational Killing vector field $\boldsymbol{K}$ that is timelike everywhere in $\mathcal{U}$, i.e., $\boldsymbol{g}(\boldsymbol{K}, \boldsymbol{K})_{| \boldsymbol{x}} > 0$ for all $\boldsymbol{x} \in \mathcal{U}$.} This detection method is different from the above quasi-local, differential geometric detection methods in so far as, instead of the local geometries of the black hole spacetimes, it relies on the effects that stationary limit surfaces have on the time evolutions of elementary spin-$\tfrac{1}{2}$ fermions and elementary spin-$0$, spin-$1$, and spin-$2$ bosons. The method is based on the observation that the Dirac, Klein--Gordon, Maxwell, and Fierz--Pauli Hamiltonians on spacelike hypersurfaces of black hole spacetimes with stationary regions undergo ellipticity-hyperbolicity transitions precisely at the locations of stationary limit surfaces, such as the outer and inner ergosurfaces of the Kerr--Newman black hole spacetime. (For the definitions of the different Hamiltonians and the notions of ellipticity and hyperbolicity see Section \ref{SectionII}.) Accordingly, one can locate stationary limit surfaces of black hole spacetimes with stationary regions by analyzing the determinantial polynomials of the principal symbols of these Hamiltonians. This is worked out in detail in Section \ref{SectionIII}, where, as an application, we also use our detection method to locate the stationary limit surfaces of the Kerr--Newman, Schwarzschild--de Sitter, and Taub--NUT black hole spacetimes. We point out that the functional analytic detection method has the advantage that its computational complexity in numerical algorithms is in general substantially lower compared to those of the differential geometric detection methods as, instead of an (at most) fourth-order nonlinear system of partial differential equations, only a quadratic polynomial of three indeterminates with possibly transcendental algebraic conditions on its coefficients has to be evaluated. Finally, since for black hole spacetimes with static regions the stationary limit surfaces coincide with the event horizons, our method can be employed, at least in this case, to quasi-locally detect event horizons. This, as well as relational concepts of event horizons and black hole entropy for black hole spacetimes with static regions, is discussed in Section \ref{SectionIV}.

\section{Preliminaries} \label{SectionII}

\noindent We first specify the differential geometric setting. We let $(\mathfrak{M}, \boldsymbol{g})$ be a smooth, connected, oriented and time-oriented Lorentzian $4$-manifold with metric signature $(1, 3, 0)$ and $\mathcal{V}$ and $\mathcal{W}$ finite rank $\mathbb{K}$-vector bundles with base space $\mathfrak{M}$ and fiberwise metrics $\boldsymbol{g}_{\mathcal{V}}$ and $\boldsymbol{g}_{\mathcal{W}}$, where $\mathbb{K} = \mathbb{R} \,\, \textnormal{or} \,\, \mathbb{C}$. Moreover, we let $\boldsymbol{\nabla}$ be the Levi--Civita connection on $T\mathfrak{M}$ and $\boldsymbol{\nabla}_{\mathcal{V}}$ and $\boldsymbol{\nabla}_{\mathcal{W}}$ the connections on $\mathcal{V}$ and $\mathcal{W}$, respectively. We assume that $(\mathfrak{M}, \boldsymbol{g})$ is stably causal, i.e., it admits a foliation by a family of spacelike hypersurfaces $(\mathfrak{N}_t)_{t \in \mathbb{R}}$, with $\mathfrak{N}_t := \{t\} \times \mathfrak{N}$, where the label $t$ is a global temporal function on $\mathfrak{M}$ \cite{MinguzziSanchez}. This assumption implies that $\mathfrak{M}$ has the product structure $\mathfrak{M} = \mathbb{R} \times \mathfrak{N}$. Consequently, if $\mathfrak{N}$ is complete, $(\mathfrak{M}, \boldsymbol{g})$ is globally hyperbolic. 

Next, we give a short account of linear, matrix-valued differential operators, their principal symbols, and their classification. For more details see, e.g., \cite{Taylor1, Evans, Hoermander} and references therein. 
\begin{Def}[Differential operator]
A linear, matrix-valued differential operator of order (at most) $m$ on $\mathfrak{M}$ is a mapping $P \hspace{-0.07cm} : \Gamma^{\infty}(\mathfrak{M}, \mathcal{V}) \rightarrow \Gamma^{\infty}(\mathfrak{M}, \mathcal{W})$ that in any local coordinate system $\boldsymbol{x}$ on an open subset $\mathcal{U} \subseteq \mathfrak{M}$ and any local trivializations $\mathcal{V}_{| \mathcal{U}} \rightarrow \mathcal{U} \times \mathbb{K}^p$ and $\mathcal{W}_{| \mathcal{U}} \rightarrow \mathcal{U} \times \mathbb{K}^q$, with $(p, q) \in \mathbb{N}^2$, is of the form
\begin{equation*}
P = \sum_{|\alpha| \leq m} \boldsymbol{A}_{\alpha}(\boldsymbol{x}) \partial_{\boldsymbol{x}}^{\alpha} \, ,
\end{equation*}
where $\Gamma^{\infty}(\mathfrak{M}, . \, )$ is the space of smooth sections over $\mathfrak{M}$, $\mathbb{N}^4_{0} \ni \alpha = (\alpha_0, ... \, , \alpha_3)$ is a multi-index of length $|\alpha| = \sum_{i = 0}^3 \alpha_i$, $\partial_{\boldsymbol{x}}^{\alpha} = \partial_{x_0}^{\alpha_0} ... \, \partial_{x_3}^{\alpha_3}$, and $\boldsymbol{A}_{\alpha} \in C^{\infty}\bigl(\mathfrak{M}, \textnormal{Mat}(q \times p, \mathbb{K})\bigr)$. 
\end{Def}
\noindent The vector space that contains these operators is usually denoted by $\textnormal{Diff}_m(\mathcal{V}, \mathcal{W})$.
\begin{Def}[Principal symbol] \label{PSymbol}
The principal symbol of a linear, matrix-valued differential operator $P \in \textnormal{Diff}_m(\mathcal{V}, \mathcal{W})$ is a bundle homomorphism $\sigma_P \hspace{-0.07cm} : \mathcal{V}_{\boldsymbol{x}} \rightarrow \mathcal{W}_{\boldsymbol{x}}$ that in any local coordinate system $\boldsymbol{x}$ on an open subset $\mathcal{U} \subseteq \mathfrak{M}$ and for any $\boldsymbol{\xi} \in T^{\star}_{\boldsymbol{x}}\mathfrak{M}$ reads 
\begin{equation} \label{PrincipalSymbol}
\sigma_P(\boldsymbol{x}, \boldsymbol{\xi}) = \sum_{|\alpha| = m} \boldsymbol{A}_{\alpha}(\boldsymbol{x}) \, \xi^{\alpha} \, .
\end{equation}
\end{Def}
\noindent This symbol captures the properties of a linear, matrix-valued differential operator that are based on its highest-order terms, thus characterizing the operator's qualitative behavior. Accordingly, one can use the principal symbol to classify linear, matrix-valued differential operators.
\begin{Def}[Ellipticity] \label{ellipticity}
A linear, matrix-valued differential operator $P \in \textnormal{Diff}_m(\mathcal{V}, \mathcal{W})$ with $p = q$ is called Petrovsky elliptic at $\boldsymbol{x} \in \mathfrak{M}$ if its principal symbol satisfies 
\begin{equation} \label{elliptic}
\textnormal{det}\bigl(\sigma_P(\boldsymbol{x}, \boldsymbol{\xi})\bigr) \not= 0
\end{equation}
for all $\boldsymbol{\xi} \in T^{\star}_{\boldsymbol{x}}\mathfrak{M} \backslash \{\boldsymbol{0}\}$. The operator is called Petrovsky elliptic in an open subset $\mathcal{U} \subseteq \mathfrak{M}$ if it is Petrovsky elliptic at every $\boldsymbol{x} \in \mathcal{U}$ and uniformly Petrovsky elliptic in $\mathcal{U}$ if there exists a constant $C > 0$ with
\begin{equation} \label{unielliptic}
\big|\textnormal{det}\bigl(\sigma_P(\boldsymbol{x}, \boldsymbol{\xi})\bigr)\big| \geq C |\boldsymbol{\xi}|^{m \cdot p}
\end{equation}
for all $\boldsymbol{x} \in \mathcal{U}$ and $\boldsymbol{\xi} \in T^{\star}_{\boldsymbol{x}}\mathfrak{M} \backslash \{\boldsymbol{0}\}$. Furthermore, writing the operator in the form $P = (P_{j k})_{j, k = 1}^p$, where $\textnormal{deg}(P_{j k}) \leq s_j + t_k$, with weights $s_j, t_k \in \mathbb{Z}$, and $P_{j k} \equiv 0$ for $s_j + t_k < 0 $, it is called Agmon--Douglis--Nirenberg elliptic or uniformly Agmon--Douglis--Nirenberg elliptic in $\mathcal{U}$ if the matrix $\bigl(\sigma_{P_{j k}}(\boldsymbol{x}, \boldsymbol{\xi})\bigr)_{j, k = 1}^p$ consisting of the principal symbols 
\begin{equation*}
\sigma_{P_{j k}}(\boldsymbol{x}, \boldsymbol{\xi}) = \sum_{|\alpha| = s_j + t_k} \bigl(\boldsymbol{A}_{\alpha}(\boldsymbol{x})\bigr)_{j k} \, \xi^{\alpha}
\end{equation*}
satisfies Inequality \textnormal{(\ref{elliptic})} or Inequality \textnormal{(\ref{unielliptic})} for every $\boldsymbol{x} \in \mathcal{U}$ and $\boldsymbol{\xi} \in T^{\star}_{\boldsymbol{x}}\mathfrak{M} \backslash \{\boldsymbol{0}\}$, respectively.
\end{Def}
\begin{Def}[Hyperbolicity] \label{hyperbolicity}
A linear, matrix-valued differential operator $P \in \textnormal{Diff}_m(\mathcal{V}, \mathcal{W})$ with $p = q$ is called Petrovsky hyperbolic at $\boldsymbol{x} \in \mathfrak{M}$ with respect to $\boldsymbol{\eta} \in T^{\star}_{\boldsymbol{x}}\mathfrak{M}\backslash\{\boldsymbol{0}\}$ if its principal symbol satisfies the conditions: 
\begin{itemize}
\item[1.] $\textnormal{det}\bigl(\sigma_P(\boldsymbol{x}, \boldsymbol{\eta})\bigr) \not = 0$.
\item[2.] The solutions $\lambda_n = \lambda_n(\boldsymbol{x}, \boldsymbol{\xi}, \boldsymbol{\eta})$, $n \in \{1, ... \, , m\}$, of the $m$th degree polynomial equation in $\lambda$ 
\begin{equation*}
\textnormal{det}\bigl(\sigma_P(\boldsymbol{x}, \boldsymbol{\xi} + \lambda \boldsymbol{\eta})\bigr) = 0	
\end{equation*}
are real-valued for every $\boldsymbol{\xi} \in T^{\star}_{\boldsymbol{x}}\mathfrak{M}\backslash\{\boldsymbol{0}\}$.
\end{itemize}
The operator is called strictly Petrovsky hyperbolic at $\boldsymbol{x} \in \mathfrak{M}$ with respect to $\boldsymbol{\eta} \in T^{\star}_{\boldsymbol{x}}\mathfrak{M}\backslash\{\boldsymbol{0}\}$ if the solutions $\lambda_n$ in Condition 2.\ are in addition pairwise distinct and (strictly) Petrovsky hyperbolic in an open subset $\mathcal{U} \subseteq \mathfrak{M}$ with respect to $\boldsymbol{\eta} \in T^{\star}_{\boldsymbol{x}}\mathfrak{M}\backslash\{\boldsymbol{0}\}$ if it is (strictly) Petrovsky hyperbolic at every $\boldsymbol{x} \in \mathcal{U}$. Moreover, it is called Agmon--Douglis--Nirenberg hyperbolic in $\mathcal{U}$ with respect to $\boldsymbol{\eta} \in T^{\star}_{\boldsymbol{x}}\mathfrak{M}\backslash\{\boldsymbol{0}\}$ if the matrix of principal symbols $\bigl(\sigma_{P_{j k}}(\boldsymbol{x}, \boldsymbol{\xi})\bigr)_{j, k = 1}^p$ introduced in Definition \ref{ellipticity} satisfies Conditions 1.\ and 2.\ for all $\boldsymbol{x} \in \mathcal{U}$ and strictly Agmon--Douglis--Nirenberg hyperbolic in $\mathcal{U}$ with respect to $\boldsymbol{\eta} \in T^{\star}_{\boldsymbol{x}}\mathfrak{M}\backslash\{\boldsymbol{0}\}$ if the solutions $\lambda_n$ in Condition 2.\ are furthermore pairwise distinct for every $\boldsymbol{x} \in \mathcal{U}$. 
\end{Def}
\noindent We point out that for general linear, matrix-valued differential operators, in particular for the Hamiltonians specified below, a proper definition of the third class, parabolicity, does not exist. However, for a certain subset of linear, matrix-valued differential operators, a suitable notion of parabolicity is presented in \cite[Chapters 15.7 and 15.8]{Taylor3}.

Lastly, we define the Dirac, Klein--Gordon, Maxwell, and Fierz--Pauli Hamiltonians, which are linear, matrix-valued differential operators of mixed elliptic-hyperbolic type, on an open subset $\mathcal{S} \subseteq \mathfrak{N}$. (Detailed information about the mathematical frameworks pertaining to these operators can be found in, e.g., the standard textbooks \cite{LawsonMichelsohn, PenroseRindler, Taylor2}.) For the purpose of defining the Dirac Hamiltonian, which determines the time evolutions of the elementary spin-$\tfrac{1}{2}$ fermions, we fix an arbitrary spin structure and let $S\mathfrak{M}$ be the corresponding spinor bundle, i.e., a complex vector bundle with fibers $S_{\boldsymbol{x}}\mathfrak{M} \simeq \mathbb{C}^4$ for all $\boldsymbol{x} \in \mathfrak{M}$. Moreover, we endow each fiber with an indefinite inner product of signature $(2, 2, 0)$  
\begin{equation} \label{SIP}
\langle \, . \, | \, . \, \rangle_{\boldsymbol{x}} \hspace{-0.07cm} : S_{\boldsymbol{x}} \mathfrak{M} \times S_{\boldsymbol{x}} \mathfrak{M} \rightarrow \mathbb{C} \, , \quad (\boldsymbol{\psi}, \boldsymbol{\chi}) \mapsto \boldsymbol{\psi}^{\dagger} \mathscr{S} \hspace{0.02cm} \boldsymbol{\chi} \, ,
\end{equation} 
and a Clifford multiplication $\boldsymbol{\gamma} \hspace{-0.07cm} : T_{\boldsymbol{x}}\mathfrak{M} \rightarrow \textnormal{End}(S_{\boldsymbol{x}}\mathfrak{M})$ that is fiber-preserving and satisfies the anti\-com\-mutation relations
\begin{equation} \label{ACR}
\bigl\{\boldsymbol{\gamma}(\boldsymbol{v}), \boldsymbol{\gamma}(\boldsymbol{w})\bigr\} = 2 \hspace{0.05cm} \boldsymbol{g}(\boldsymbol{v}, \boldsymbol{w}) \, \1_{S_{\boldsymbol{x}}\mathfrak{M}}
\end{equation}
for all $\boldsymbol{x} \in \mathfrak{M}$ and $\boldsymbol{v}, \boldsymbol{w} \in T_{\boldsymbol{x}}\mathfrak{M}$. The matrix $\mathscr{S}$ in the indefinite inner product (\ref{SIP}) is chosen such that $\boldsymbol{\gamma}^{\dagger} \circ \mathscr{S} = \mathscr{S} \circ \boldsymbol{\gamma}$. Then, the Dirac operator $\mathcal{D} \hspace{-0.07cm} : \Gamma^{\infty}(\mathfrak{M}, S\mathfrak{M}) \rightarrow \Gamma^{\infty}(\mathfrak{M}, S\mathfrak{M})$ with an external potential $\mathcal{B} \in \textnormal{End}(S\mathfrak{M})$ in any local coordinate system $\boldsymbol{x}$ on an open subset $\mathcal{U} \subseteq \mathfrak{M}$ and in any local trivialization $S\mathfrak{M}_{| \mathcal{U}} \rightarrow \mathcal{U} \times \mathbb{C}^4$ takes the form
\begin{equation*} 
\textnormal{Diff}_1(S\mathfrak{M}, S\mathfrak{M}) \ni \mathcal{D} = i \gamma^{\mu} \nabla_{\mu} + \mathcal{B}(\boldsymbol{x}) \, , \quad \mu \in \{0, 1, 2, 3\} \, ,
\end{equation*}
where $\nabla_{\mu}$ are the components of the covariant derivative corresponding to the unique metric connection on $S\mathfrak{M}$ obtained as a lift of the Levi--Civita connection $\boldsymbol{\nabla}$ on $T\mathfrak{M}$. Using this operator, we can define the Dirac equation as
\begin{equation*} 
(\mathcal{D} - m) \boldsymbol{\psi}(\boldsymbol{x}) = \boldsymbol{0} \, ,
\end{equation*}
in which the Dirac wave function $\boldsymbol{\psi} \in \Gamma^{\infty}_{\textnormal{sc}}(\mathfrak{M}, S\mathfrak{M})$, the subscript ``sc'' denoting spatially compact support, and $m \in \mathbb{R}_{\geq 0}$ is the rest mass of the spin-$\tfrac{1}{2}$ fermion associated with the Dirac wave function. Choosing local coordinates $\boldsymbol{x} = (t, \boldsymbol{y})$ on $\mathbb{R} \times \mathcal{S} \subseteq \mathfrak{M}$ and multiplying this equation by the inverse of $\gamma^t$, separation of the partial time derivative term yields the Dirac equation in Hamiltonian form 
\begin{equation} \label{DEH}
i \partial_t \boldsymbol{\psi}(t, \boldsymbol{y}) = \biggl[(\gamma^t)^{- 1} \bigl(- i \gamma^j \nabla_j - \mathcal{B} + m\bigr) + \frac{\omega_{t \mu \nu}}{8 i} \bigl[\gamma^{\mu}, \gamma^{\nu}\bigr]\biggr] \boldsymbol{\psi}(t, \boldsymbol{y}) =: \mathcal{H}_{\mathcal{D}} \boldsymbol{\psi}(t, \boldsymbol{y}) \, , \quad j \in \{1, 2, 3\} \, ,
\end{equation}
where $\mathcal{H}_{\mathcal{D}} \hspace{-0.07cm} : \Gamma^{\infty}(\mathfrak{N}, S\mathfrak{M}) \rightarrow \Gamma^{\infty}(\mathfrak{N}, S\mathfrak{M})$ is the Dirac Hamiltonian and $\omega_{t \mu \nu}$ are the temporal components of the connection $1$-forms $\omega_{\mu \nu} = \boldsymbol{g}(\boldsymbol{\nabla} \boldsymbol{e}_{\mu}, \boldsymbol{e}_{\nu})$ of the Levi--Civita connection with respect to a local orthonormal frame $(\boldsymbol{e}_{\mu})_{\mu \in \{0, 1, 2, 3\}}$. For time-independent spacetimes, the Dirac Hamiltonian with domain of definition $\textnormal{Dom}(\mathcal{H}_{\mathcal{D}}) = \Gamma^{\infty}_{\textnormal{c}}(\mathfrak{N} \backslash \partial \mathfrak{N}, S\mathfrak{M})$, the subscript ``c'' signifying compact support, is formally self-adjoint with respect to the scalar product
\begin{equation*}
(\boldsymbol{\psi} \hspace{0.02cm} | \hspace{0.02cm} \boldsymbol{\chi})_{\mathfrak{N}} = \int_{\mathfrak{N}} \langle \boldsymbol{\psi} \hspace{0.02cm} | \hspace{0.02cm} \gamma^{\mu} \nu_{\mu} \hspace{0.02cm} \boldsymbol{\chi}\rangle_{\boldsymbol{x}} \, \textnormal{d}\mu_{\mathfrak{N}}(\boldsymbol{x}) \, ,
\end{equation*}
in which $\boldsymbol{\nu}$ is a future-directed, timelike normal on $\mathfrak{N}$ and $\textnormal{d}\mu_{\mathfrak{N}}$ the invariant measure on $(\mathfrak{N}, \boldsymbol{g}_{\mathfrak{N}})$, $\boldsymbol{g}_{\mathfrak{N}}$ being the induced metric on $\mathfrak{N}$. We note that in this particular framework for the Dirac equation, a Dirac wave function $\boldsymbol{\psi}(t, . \, )$ is regarded, for every value of the time coordinate $t$, as being contained in the completion of the Hilbert space $\bigl(\Gamma^{\infty}_{\textnormal{c}}(\mathfrak{N} \backslash \partial \mathfrak{N}, S\mathfrak{M}), ( \, . \, | \, . \, )_{\mathfrak{N}}\bigr)$. 

In order to define the Klein--Gordon, Maxwell, and Fierz--Pauli Hamiltonians, which determine the time evolutions of the elementary spin-$0$, spin-$1$, and spin-$2$ bosons, respectively, we employ the d'Alembert operator $\square \hspace{-0.07cm} : \Gamma^{\infty}(\mathfrak{M}, \mathfrak{M} \times \mathbb{K}) \rightarrow \Gamma^{\infty}(\mathfrak{M}, \mathfrak{M} \times \mathbb{K})$ with an external potential $\mathfrak{B} \in \textnormal{End}(\mathfrak{M} \times \mathbb{K})$ that in any local coordinate system $\boldsymbol{x}$ on $\mathcal{U}$ and in any local trivialization $\mathfrak{M} \times \mathbb{K}_{| \mathcal{U}} \rightarrow \mathcal{U} \times \mathbb{K}$ reads 
\begin{equation*} 
\textnormal{Diff}_2(\mathfrak{M} \times \mathbb{K}, \mathfrak{M} \times \mathbb{K}) \ni \square = \frac{1}{\sqrt{|\textnormal{det}(\boldsymbol{g})|}} \, \partial_{\mu} \bigl(\sqrt{|\textnormal{det}(\boldsymbol{g})|} \,\, g^{\mu \nu} \partial_{\nu}\bigr) + \mathfrak{B}(\boldsymbol{x}) \, .
\end{equation*}
With this operator, we can, successively, define the Klein--Gordon equation as
\begin{equation} \label{KGE} 
(\square - \mathfrak{m}^2) \phi(\boldsymbol{x}) = 0 \, ,
\end{equation}
where the scalar wave function $\phi \in \Gamma^{\infty}_{\textnormal{sc}}(\mathfrak{M}, \mathfrak{M} \times \mathbb{C})$ and $\mathfrak{m} \in \mathbb{R}_{\geq 0}$ is the rest mass of the spin-$0$ boson associated with the scalar wave function, the (homogeneous) electromagnetic wave equation as
\begin{equation} \label{EWE}
\square A_{\mu}(\boldsymbol{x}) = R_{\mu \nu} A^{\nu}(\boldsymbol{x}) \, ,
\end{equation}
in which $A_{\mu} \in \Gamma^{\infty}_{\textnormal{sc}}(\mathfrak{M}, \mathfrak{M} \times \mathbb{R})$ are the components of the electromagnetic $4$-potential subject to the Lorenz gauge condition $\nabla^{\mu} A_{\mu} = 0$ and $R_{\mu \nu}$ are the usual Ricci curvature tensor components,\footnote{Except for an additional zero-order term, the same electromagnetic wave equation holds for the inhomogeneous case.} and the Fierz--Pauli equation\footnote{This equation describes the propagation of (linear) gravitons on curved background spacetimes \cite{AragoneDeser, BuchbinderEtAl}.} as
\begin{equation} \label{FPE}
\square H_{\mu \nu}(\boldsymbol{x}) = - 2 R_{\alpha \mu \beta \nu} H^{\alpha \beta}(\boldsymbol{x}) \, ,
\end{equation}
with the graviton field components $H_{\mu \nu} \in \Gamma^{\infty}_{\textnormal{sc}}(\mathfrak{M}, \mathfrak{M} \times \mathbb{R})$ satisfying the Lorenz gauge condition $\nabla^{\mu} H_{\mu \nu} = 0$ and the trace condition $H\indices{^{\mu}_{\mu}} = 0$, and where $R_{\alpha \mu \beta \nu}$ are the components of the Riemann curvature tensor. In local coordinates $\boldsymbol{x} = (t, \boldsymbol{y})$ on $\mathbb{R} \times \mathcal{S} \subseteq \mathfrak{M}$, the Hamiltonian form of the Klein--Gordon equation (\ref{KGE}) is given by the system
\begin{equation} \label{KGEH}
i \partial_t \Phi(t, \boldsymbol{y}) = \frac{1}{g^{t t}} \hspace{-0.05cm}
\begin{pmatrix} 
0 & g^{t t} \\ 
\varepsilon & - i \kappa 
\end{pmatrix} \hspace{-0.07cm} \Phi(t, \boldsymbol{y}) =: \mathcal{H}_{\textnormal{KG}} \Phi(t, \boldsymbol{y}) \, ,
\end{equation}
with the Klein--Gordon Hamiltonian $\mathcal{H}_{\textnormal{KG}} \hspace{-0.07cm} : \Gamma^{\infty}(\mathfrak{N}, \mathfrak{M} \times \mathbb{C}) \rightarrow \Gamma^{\infty}(\mathfrak{N}, \mathfrak{M} \times \mathbb{C})$, the wave function $\Phi(t, \boldsymbol{y}) := (1, i \partial_t)^{\textnormal{T}} \hspace{0.02cm} \phi(t, \boldsymbol{y})$, and 
\begin{equation*}
\begin{split}
\varepsilon & := g^{i j} \partial_i \partial_j + \frac{1}{\sqrt{|\textnormal{det}(\boldsymbol{g})|}} \, \partial_{\mu} \bigl(\sqrt{|\textnormal{det}(\boldsymbol{g})|} \, g^{\mu j}\bigr) \partial_j + \mathfrak{B} - \mathfrak{m}^2 \\[0.25cm]
\kappa & := 2 g^{t j} \partial_j + \frac{1}{\sqrt{|\textnormal{det}(\boldsymbol{g})|}} \, \partial_{\mu} \bigl(\sqrt{|\textnormal{det}(\boldsymbol{g})|} \, g^{\mu t}\bigr) \, .
\end{split}
\end{equation*}
Taking the domain of definition of the Klein--Gordon Hamiltonian on time-independent spacetimes to be $\textnormal{Dom}(\mathcal{H}_{\textnormal{KG}}) = \Gamma^{\infty}_{\textnormal{c}}(\mathfrak{N} \backslash \partial \mathfrak{N}, \mathfrak{M} \times \mathbb{C}^2)$, this operator is formally self-adjoint with respect to the inner product
\begin{equation*}
\Sl\Phi \hspace{0.02cm} | \hspace{0.02cm} \Psi\Sr_{\mathfrak{N}} = \int_{\mathfrak{N}} \Phi^{\dagger} \hspace{-0.07cm}
\begin{pmatrix} 
\varepsilon/g^{t t} & 0 \\ 
0 & 1
\end{pmatrix} 
\hspace{-0.05cm} \Psi \, \textnormal{d}\mu_{\mathfrak{N}}(\boldsymbol{x}) \, .
\end{equation*}
Similar Hamiltonian forms and inner products result for the (homogeneous) electromagnetic wave equation (\ref{EWE}) and the Fierz--Pauli equation (\ref{FPE}). We point out that the above Hamiltonians comprise, inter alia, the spatial parts of either the Dirac or the d'Alembert operator on $\mathfrak{M}$ (times the inverse of $\gamma^t$ or the reciprocal of $g^{t t}$, respectively), obtained by separating the partial time derivative term, and not the intrinsic Dirac or d'Alembert operators on $\mathfrak{N}$, which, unlike the former, are elliptic everywhere on $\mathfrak{N}$. This fact is of paramount importance in the analysis of the determinantial polynomials of the principal symbols of the Hamiltonians in the next section.

\section{Quasi-local, Functional Analytic Detection Method for Stationary Limit Surfaces} \label{SectionIII}

\subsection{Derivation} \label{SectionIIIA}

\noindent To locate the stationary limit surfaces of black hole spacetimes with stationary regions through ellipticity-hyperbolicity transitions of the Hamiltonians of the elementary fermions and bosons, the behaviors of the determinantial polynomials of the principal symbols of these operators have to be analyzed. Thus, we begin by computing the principal symbol [see Equation (\ref{PrincipalSymbol}) in Definition \ref{PSymbol}] of the Dirac Hamiltonian $\mathcal{H}_{\mathcal{D}}$ defined in Equation~(\ref{DEH}) at a point $\boldsymbol{y} \in \mathfrak{N}$ of a spacelike hypersurface of a black hole spacetime with stationary regions, which yields 
\begin{equation*}
\sigma_{\mathcal{H}_{\mathcal{D}}}(\boldsymbol{y}, \boldsymbol{\xi}) = - i \hspace{0.04cm} (\gamma^t)^{- 1} \gamma^j \xi_j \, ,
\end{equation*}
with $\boldsymbol{\xi} \in T^{\star}_{\boldsymbol{y}} \mathfrak{N} \backslash \{\boldsymbol{0}\}$. Making use of the relations
\begin{equation*}
(\gamma^t)^{- 1} \, (\gamma^t)^{- 1} = \frac{\1_{S_{\boldsymbol{x}}\mathfrak{M}}}{g^{t t}} \quad \textnormal{and} \quad \gamma^j \xi_j \, \gamma^k \xi_k = g^{j k} \xi_j \xi_k \, \1_{S_{\boldsymbol{x}}\mathfrak{M}} \, ,
\end{equation*}
which follow directly from Equation (\ref{ACR}), the determinant of this principal symbol is given by 
\begin{equation} \label{DetPolDir}
\textnormal{det}\bigl(\sigma_{\mathcal{H}_{\mathcal{D}}}(\boldsymbol{y}, \boldsymbol{\xi})\bigr) = (- i)^4 \, \textnormal{det}\bigl((\gamma^t)^{- 1}\bigr) \, \textnormal{det}(\gamma^j \xi_j) = \biggl[\frac{g^{j k} \xi_j \xi_k}{g^{t t}}\biggr]^2 \, . 
\end{equation}
The principal symbol of the Klein--Gordon Hamiltonian $\mathcal{H}_{\textnormal{KG}}$ specified in Equation (\ref{KGEH}), on the other hand, is singular and nondiagonizable, this being an artifact of the reduction of order in going from the second-order equation (\ref{KGE}) to the first-order system (\ref{KGEH}). Therefore, we evaluate the matrix of principal symbols of the components of the Klein--Gordon Hamiltonian as in the notions of Agmon--Douglis--Nirenberg ellipticity and hyperbolicity (cf.\ Definitions \ref{ellipticity} and \ref{hyperbolicity}), resulting in
\begin{equation*} 
\bigl(\sigma_{(\mathcal{H}_{\textnormal{KG}})_{j k}}(\boldsymbol{y}, \boldsymbol{\xi})\bigr)_{j, k = 1}^2 = \frac{1}{g^{t t}} \hspace{-0.05cm}
\begin{pmatrix} 
0 & g^{t t} \\ 
g^{j k} \xi_j \xi_k & - 2 i \hspace{0.03cm} g^{t k} \xi_k 
\end{pmatrix} ,
\end{equation*}
where the corresponding determinant reads
\begin{equation} \label{DetPolDalem}
\textnormal{det}\Bigl(\bigl(\sigma_{(\mathcal{H}_{\textnormal{KG}})_{j k}}(\boldsymbol{y}, \boldsymbol{\xi})\bigr)_{j, k = 1}^2\Bigr) = - \frac{g^{j k} \xi_j \xi_k}{g^{t t}} \, .
\end{equation}
These computations show that both the Dirac Hamiltonian and the Klein--Gordon Hamiltonian (and hence the Maxwell and Fierz--Pauli Hamiltonians) are elliptic at $\boldsymbol{y} \in \mathfrak{N}$ if 
\begin{equation} \label{EC}
g^{j k} \xi_j \xi_k \not= 0 
\end{equation}
for all $\boldsymbol{\xi} \in T^{\star}_{\boldsymbol{y}} \mathfrak{N} \backslash \{\boldsymbol{0}\}$ [see Inequality (\ref{elliptic}) in Definition \ref{ellipticity}]. However, if this is not the case and if for some $\boldsymbol{\eta} \in T^{\star}_{\boldsymbol{y}} \mathfrak{N} \backslash \{\boldsymbol{0}\}$ satisfying Inequality (\ref{EC}) the two zeros of the second-order polynomial in $\lambda$
\begin{equation*} 
g^{j k} \bigl[\lambda^2 \eta_j \eta_k + 2 \lambda \, \eta_j \xi_k + \xi_j \xi_k\bigr]
\end{equation*}
are real-valued for every $\boldsymbol{\xi} \in T^{\star}_{\boldsymbol{y}} \mathfrak{N} \backslash \{\boldsymbol{0}\}$, which is determined by the condition 
\begin{equation} \label{RZ}
\bigl(g^{j k} \eta_j \xi_k\bigr)^2 \geq \bigl(g^{j k} \eta_j \eta_k\bigr) \bigl(g^{l m} \xi_l \xi_m\bigr) \, ,
\end{equation}
the Hamiltonians are hyperbolic at $\boldsymbol{y} \in \mathfrak{N}$ with respect to said $\boldsymbol{\eta}$ instead (see Condition 2.\ in Definition \ref{hyperbolicity}). At an ellipticity-hyperbolicity transition at $\boldsymbol{y} \in \mathfrak{N}$, where the Hamiltonians are neither elliptic nor hyperbolic with respect to any $\boldsymbol{\eta} \in T^{\star}_{\boldsymbol{y}} \mathfrak{N} \backslash \{\boldsymbol{0}\}$, there exists at least one $\boldsymbol{\xi} \in T^{\star}_{\boldsymbol{y}} \mathfrak{N} \backslash \{\boldsymbol{0}\}$ such that 
\begin{equation} \label{SLSCN}
g^{j k} \xi_j \xi_k = 0 \, ,
\end{equation}
and for those $\boldsymbol{\eta} \in T^{\star}_{\boldsymbol{y}} \mathfrak{N} \backslash \{\boldsymbol{0}\}$ for which Inequality (\ref{EC}) holds, Inequality (\ref{RZ}) is not satisfied for all $\boldsymbol{\xi} \in T^{\star}_{\boldsymbol{y}} \mathfrak{N} \backslash \{\boldsymbol{0}\}$. 

Consequently, as the Hamiltonians are elliptic in stationary spacetime regions and hyperbolic in nonstationary spacetime regions, the simultaneous violations of Inequalities (\ref{EC}) and (\ref{RZ}), which, taken together, are sufficient conditions for the localization of ellipticity-hyperbolicity transitions, can be used to detect the stationary limit surfaces of a given black hole spacetime with stationary regions.\footnote{Since the purely spatial part of the inverse spacetime metric $\boldsymbol{g}^{- 1}$ is in general indefinite, this can also be seen from the fact that Equation (\ref{SLSCN}) acts as a stationary limit condition on $\mathfrak{N}$.}\fnsep\footnote{Because of their generality, Inequalities (\ref{EC}) and (\ref{RZ}) are not limited to black hole spacetimes. Accordingly, they can be employed to locate the stationary limit surfaces of any spacetime having the product structure $\mathfrak{M} = \mathbb{R} \times \mathfrak{N}$.} This functional analytic detection method for stationary limit surfaces is quasi-local in the sense that it employs solely operators on $\mathfrak{N}$. Therefore, it locates stationary limit surfaces only on the spacelike hypersurfaces of black hole spacetimes where these operators are defined. Moreover, since the method is based on their principal symbols, it is independent of any particular choice of coordinates on $\mathfrak{N}$, or rather on the associated patch of $\mathfrak{M}$ as long as one uses coordinate systems for which the Hamiltonians are time-independent and coordinate transformations that are diffeomorphisms from $\mathfrak{N}$ to itself. Working with coordinate systems that give rise to time-dependent Hamiltonians, such as Kruskal--Szekeres-type coordinates or Lema\^itre-type coordinates,\footnote{For the original Kruskal--Szekeres and Lema\^itre coordinates see \cite{Kruskal, Szekeres} and \cite{Lemaitre}, respectively, and for certain generalizations thereof \cite[Chapter 6.58.]{ChandraBook} and \cite{Sorge}.} would, however, break the coordinate invariance and may result in different classifications of the Hamiltonians. Hence, being limited to coordinate systems on $\mathfrak{M}$ for which the purely spatial part $(g^{i j})_{i, j \in \{1, 2, 3\}}$ of the inverse spacetime metric $\boldsymbol{g}^{- 1}$ is time-independent,\footnote{This directly follows from Equations (\ref{DetPolDir}) and (\ref{DetPolDalem}) together with the time-independence of the Hamiltonians.} Inequalities (\ref{EC}) and (\ref{RZ}) can, instead of $(g^{i j})_{i, j \in \{1, 2, 3\}}$, be evaluated by the inverse induced metric $\boldsymbol{g}^{- 1}_{\mathfrak{N}} = \phi^{\star} \boldsymbol{g}^{- 1}$ on the spacelike hypersurface $\mathfrak{N}$, where $\phi \hspace{-0.07cm} : \mathfrak{N} \hookrightarrow \mathfrak{M}$ is the trivial isometric embedding and the asterisk denotes the pullback from $\mathfrak{M}$ to $\mathfrak{N}$. This aspect further illustrates the quasi-locality of the functional analytic detection method.

\subsection{Application} \label{SectionIIIB}

\noindent We now locate the stationary limit surfaces of the Kerr--Newman, Schwarzschild--de Sitter, and Taub--NUT black hole spacetimes using our functional analytic detection method. More precisely, for each case study, we determine and examine the violations of Inequalities (\ref{EC}) and (\ref{RZ}), thus identifying the locations where the Hamiltonians undergo ellipticity-hyperbolicity transitions.
\begin{Example}[Kerr--Newman black hole spacetime] \label{Ex1}
\noindent To properly represent the axially symmetric Kerr--Newman black hole spacetime \cite{KerrNewman}, which can be employed to describe the gravitational field of an isolated, rotating, electrically charged black hole, we work with an advanced Eddington--Finkelstein coordinate system $(\tau, r, \vartheta, \varphi) \in \mathbb{R} \times \mathbb{R}_{> 0} \times (0, \pi) \times [0, 2 \pi)$ (this coordinate system is a direct generalization of the advanced Eddington--Finkelstein coordinate system used in \cite{Roeken} for the Kerr black hole spacetime to the Kerr--Newman case). Expressed in these coordinates, the nonextreme Kerr--Newman metric is time-independent and regular everywhere, reading 
\begin{equation} \label{KNmetric}
\begin{split}	
\boldsymbol{g} & = \biggl[1 - \frac{2 M r - Q^2}{\Sigma}\biggr] \textnormal{d}\tau \otimes \textnormal{d}\tau - \frac{2 M r - Q^2}{\Sigma} \, \bigl(\textnormal{d}\tau \otimes \bigl[\textnormal{d}r - a \sin^2(\vartheta) \, \textnormal{d}\varphi\bigr] + \bigl[\textnormal{d}r - a \sin^2(\vartheta) \, \textnormal{d}\varphi\bigr] \otimes \textnormal{d}\tau\bigr) \\[0.25cm]
& \hspace{0.45cm} - \biggl[1 + \frac{2 M r - Q^2}{\Sigma}\biggr] \big(\textnormal{d}r - a \sin^2(\vartheta) \, \textnormal{d}\varphi\bigr) \otimes \big(\textnormal{d}r - a \sin^2(\vartheta) \, \textnormal{d}\varphi\bigr) - \Sigma \, \boldsymbol{g}_{S^2} \, ,
\end{split}
\end{equation}
where $M \in \mathbb{R}_{> 0}$ is the (ADM) mass, $a \in \mathbb{R}$ is the angular momentum per unit mass, and $Q \in \mathbb{R}$ is the charge of the Kerr--Newman black hole satisfying $||(a, Q)||_{\mathbb{R}^2} < M$. Moreover, $\Sigma = \Sigma(r, \vartheta) := r^2 + a^2 \cos^2(\vartheta)$ and $\boldsymbol{g}_{S^2} = \textnormal{d}\vartheta \otimes \textnormal{d}\vartheta + \sin^2(\vartheta) \, \textnormal{d}\varphi \otimes \textnormal{d}\varphi$ is the metric on the unit $2$-sphere. Substituting the corresponding inverse induced metric on $\mathfrak{N}$ given by
\begin{equation*}		
\boldsymbol{g}^{- 1}_{\mathfrak{N}} = - \frac{1}{\Sigma} \bigl([r^2 - 2 M r + a^2 + Q^2] \, \partial_r \otimes \partial_r + a \, (\partial_r \otimes \partial_{\varphi} + \partial_{\varphi} \otimes \partial_r) + \boldsymbol{g}^{- 1}_{S^2}\bigr)
\end{equation*}
into Inequalities (\ref{EC}) and (\ref{RZ}), we obtain, after simple algebraic manipulations, 
\begin{equation} \label{KNI1} 
[r^2 - 2 M r + a^2 \cos^2(\vartheta)+ Q^2] \, \xi^2_r + [a \sin(\vartheta) \, \xi_r + \textnormal{csc}(\vartheta) \, \xi_{\varphi}]^2 + \xi^2_{\vartheta} \not= 0 
\end{equation}
and
\begin{equation} \label{KNI2} 
\begin{split}
& - \bigl[r^2 - 2 M r + a^2 \cos^2(\vartheta) + Q^2\bigr] \bigl[\sin^2(\vartheta) \, (\eta_r \xi_{\vartheta} - \eta_{\vartheta} \xi_r)^2 + (\eta_r \xi_{\varphi} - \eta_{\varphi} \xi_r)^2\bigr] \\[0.25cm] 
& \geq \bigl[\eta_{\vartheta} \xi_{\varphi} - \eta_{\varphi} \xi_{\vartheta} - a \sin^2(\vartheta) \, (\eta_r \xi_{\vartheta} - \eta_{\vartheta} \xi_r)\bigr]^2 \, ,
\end{split}
\end{equation}
respectively. Considering all $\boldsymbol{\xi}, \boldsymbol{\eta} \in T^{\star}_{\boldsymbol{y}} \mathfrak{N} \backslash \{\boldsymbol{0}\}$, we immediately find that these inequalities are not ful\-filled---and therefore that the Hamiltonians undergo ellipticity-hyperbolicity transitions---at
\begin{equation*}
r_{1, 2} = M \pm \sqrt{M^2  - a^2 \cos^2(\vartheta) - Q^2} =: r^{\pm}_{\mathcal{E}} \, , 
\end{equation*}
which are the radial locations of the outer and inner ergosurfaces. Hence, the outer and inner ergosurfaces at $\mathbb{R} \times \{r^{\pm}_{\mathcal{E}}\} \times (0, \pi) \times [0, 2 \pi)$ are the stationary limit surfaces of the Kerr--Newman black hole spacetime. Furthermore, the Hamiltonians are elliptic outside the ergosphere in $(0, r^-_{\mathcal{E}}) \cup (r^+_{\mathcal{E}}, \infty) \times (0, \pi) \times [0, 2 \pi)$ and hyperbolic inside the ergosphere in $(r^-_{\mathcal{E}}, r^+_{\mathcal{E}}) \times (0, \pi) \times [0, 2 \pi)$ with respect to any $\boldsymbol{\eta} \in T^{\star}_{\boldsymbol{y}} \mathfrak{N} \backslash \{\boldsymbol{0}\}$ satisfying Inequality (\ref{KNI1}) and Inequality (\ref{KNI2}) for all $\boldsymbol{\xi} \in T^{\star}_{\boldsymbol{y}} \mathfrak{N} \backslash \{\boldsymbol{0}\}$. It is important to stress that the ellipticity-hyperbolicity transitions of the Hamiltonians at the outer and inner ergosurfaces are not artifacts of coordinate singularities, as the nonextreme Kerr--Newman metric in the advanced Eddington--Finkelstein coordinate representation (\ref{KNmetric}) is regular on the entire domain of the black hole spacetime.
\end{Example}	
\begin{Example}[Schwarzschild--de Sitter black hole spacetime] \label{SdS}
\noindent In the case of the Schwarzschild--de Sitter black hole spacetime \cite{Weyl}, which models a spherically symmetric black hole in the center of an inflating universe, we represent the Schwarzschild--de Sitter metric also in advanced Eddington--Finkelstein coordinates $(\tau, r, \vartheta, \varphi) \in \mathbb{R} \times \mathbb{R}_{> 0} \times (0, \pi) \times [0, 2 \pi)$, resulting in the time-independent and regular expression
\begin{equation*} 
\boldsymbol{g} = \biggl[1 - \frac{2 M}{r} - \frac{\Lambda r^2}{3}\biggr] \textnormal{d}\tau \otimes \textnormal{d}\tau - \biggl[\frac{2 M}{r} + \frac{\Lambda r^2}{3}\biggr] (\textnormal{d}\tau \otimes \textnormal{d}r + \textnormal{d}r \otimes \textnormal{d}\tau) - \biggl[1 + \frac{2 M}{r} + \frac{\Lambda r^2}{3}\biggr] \textnormal{d}r \otimes \textnormal{d}r - r^2 \, \boldsymbol{g}_{S^2} \, .
\end{equation*}
Here, $M \in \mathbb{R}_{> 0}$ is again the (ADM) mass of the black hole and $\Lambda \in \mathbb{R}_{> 0}$ is the cosmological constant. Using the corresponding inverse induced metric on $\mathfrak{N}$
\begin{equation*}		
\boldsymbol{g}^{- 1}_{\mathfrak{N}} = - \biggl[1 - \frac{2 M}{r} - \frac{\Lambda r^2}{3}\biggr] \, \partial_r \otimes \partial_r - \frac{1}{r^2} \, \boldsymbol{g}^{- 1}_{S^2} \, ,
\end{equation*}
Inequalities (\ref{EC}) and (\ref{RZ}) readily become
\begin{equation} \label{SdSI1} 
r \biggl[r - 2 M - \frac{\Lambda r^3}{3}\biggr] \xi^2_r + \xi^2_{\vartheta} + \textnormal{csc}^2(\vartheta) \, \xi^2_{\varphi} \not= 0 
\end{equation}
and 
\begin{equation} \label{SdSI2} 
- r \biggl[r - 2 M - \frac{\Lambda r^3}{3}\biggr] \bigl[\sin^2(\vartheta) \, (\eta_r \xi_{\vartheta} - \eta_{\vartheta} \xi_r)^2 + (\eta_r \xi_{\varphi} - \eta_{\varphi} \xi_r)^2\bigr] \geq (\eta_{\vartheta} \xi_{\varphi} - \eta_{\varphi} \xi_{\vartheta})^2 \, .
\end{equation}
Thus, the locations of the ellipticity-hyperbolicity transitions of the Hamiltonians, taking into account all $\boldsymbol{\xi}, \boldsymbol{\eta} \in T^{\star}_{\boldsymbol{y}} \mathfrak{N} \backslash \{\boldsymbol{0}\}$, are given by the zeros of the depressed cubic polynomial $r^3 - 3 r/\Lambda + 6 M/\Lambda$ where both of these inequalities do not hold. In more detail, for the physically relevant case $9 M^2 \Lambda \in (0, 1)$, there exist, on the one hand, two distinct positive, real-valued zeros at 
\begin{equation*}
r_1 = \frac{2}{\sqrt{\Lambda}} \cos\biggl(\frac{1}{3} \arccos\bigl(- 3 M \sqrt{\Lambda} \, \bigr) + \frac{4 \pi}{3}\biggr) =: r_{\textnormal{EH}}	
\end{equation*}
and 
\begin{equation*}
r_2 = \frac{2}{\sqrt{\Lambda}} \cos\biggl(\frac{1}{3} \arccos\bigl(- 3 M \sqrt{\Lambda} \, \bigr)\biggr) =: r_{\textnormal{CH}}
\end{equation*}
with $r_1 < r_2$, where the smaller zero corresponds to the Schwarzschild-type black hole event horizon and the larger zero to the de Sitter-type cosmological event horizon, and, on the other hand, a third negative, real-valued zero 
\begin{equation*}
r_3 = \frac{2}{\sqrt{\Lambda}} \cos\biggl(\frac{1}{3} \arccos\bigl(- 3 M \sqrt{\Lambda} \, \bigr) + \frac{2 \pi}{3}\biggr) \, ,
\end{equation*}
which, however, lies outside the range of the radial coordinate. Accordingly, the two event horizons at $\mathbb{R} \times \{r_{\textnormal{EH, CH}}\} \times (0, \pi) \times [0, 2 \pi)$ are the stationary limit surfaces of the Schwarzschild--de Sitter black hole spacetime. Besides, the Hamiltonians are elliptic in the region $(r_{\textnormal{EH}}, r_{\textnormal{CH}}) \times (0, \pi) \times [0, 2 \pi)$ between the horizons and hyperbolic inside the Schwarzschild-type black hole event horizon and outside the de Sitter-type cosmological event horizon in $(0, r_{\textnormal{EH}}) \cup (r_{\textnormal{CH}}, \infty) \times (0, \pi) \times [0, 2 \pi)$ with respect to any $\boldsymbol{\eta} \in T^{\star}_{\boldsymbol{y}} \mathfrak{N} \backslash \{\boldsymbol{0}\}$ fulfilling both Inequality (\ref{SdSI1}) and Inequality (\ref{SdSI2}) for all $\boldsymbol{\xi} \in T^{\star}_{\boldsymbol{y}} \mathfrak{N} \backslash \{\boldsymbol{0}\}$.
\end{Example}	
\begin{Example}[Taub--NUT black hole spacetime]
\noindent Finally, for describing the Taub--NUT black hole spacetime \cite{Taub, NUT}, which is usually considered as a simple rotating generalization of the standard Schwarzschild black hole spacetime, we use advanced Eddington--Finkelstein-type coordinates $(\tau, r, \vartheta, \varphi) \in \mathbb{R} \times \mathbb{R} \times (0, \pi) \times [0, 2 \pi)$ (see, e.g., \cite[Chapter 12.1.2]{GriffithsPodolsky} and \cite{KagramanovaEtAl}). In these coordinates, the Taub--NUT metric is time-independent and regular on the entire black hole spacetime and can be written as   
\begin{equation*}
\begin{split}	
\boldsymbol{g} & = 4 l^2 \biggl[1 - \frac{2 (M r + l^2)}{r^2 + l^2}\biggr] \bigl(\textnormal{d}\tau - \cos(\vartheta) \, \textnormal{d}\varphi\bigr) \otimes \bigl(\textnormal{d}\tau - \cos(\vartheta) \, \textnormal{d}\varphi\bigr) - 2 l \bigl[\bigl(\textnormal{d}\tau - \cos(\vartheta) \, \textnormal{d}\varphi\bigr) \otimes \textnormal{d}r \\[0.25cm]
& \hspace{0.45cm} + \textnormal{d}r \otimes \bigl(\textnormal{d}\tau - \cos(\vartheta) \, \textnormal{d}\varphi\bigr)\bigr] - [r^2 + l^2] \, \boldsymbol{g}_{S^2} \, ,
\end{split}
\end{equation*}
where $M \in \mathbb{R}_{> 0}$ is once again interpreted as the (ADM) mass of the black hole and the so-called NUT parameter $l \in \mathbb{R}$ can be regarded as a measure of its angular momentum. Then, inserting the corresponding inverse induced metric on $\mathfrak{N}$
\begin{equation*}		
\boldsymbol{g}^{- 1}_{\mathfrak{N}} = - \biggl[1 - \frac{2 (M r + l^2)}{r^2 + l^2}\biggr] \, \partial_r \otimes \partial_r - \frac{1}{r^2 + l^2} \, \boldsymbol{g}^{- 1}_{S^2} 
\end{equation*}
into Inequalities (\ref{EC}) and (\ref{RZ}) yields 
\begin{equation} \label{TNI1} 
[r^2 - 2 M r - l^2] \, \xi^2_r + \xi^2_{\vartheta} + \textnormal{csc}^2(\vartheta) \, \xi^2_{\varphi} \not= 0 
\end{equation}
and 
\begin{equation} \label{TNI2} 
- [r^2 - 2 M r - l^2] \, \bigl[\sin^2(\vartheta) \, (\eta_r \xi_{\vartheta} - \eta_{\vartheta} \xi_r)^2 + (\eta_r \xi_{\varphi} - \eta_{\varphi} \xi_r)^2\bigr] \geq (\eta_{\vartheta} \xi_{\varphi} - \eta_{\varphi} \xi_{\vartheta})^2 \, .
\end{equation}
Evaluating these inequalities for all $\boldsymbol{\xi}, \boldsymbol{\eta} \in T^{\star}_{\boldsymbol{y}} \mathfrak{N} \backslash \{\boldsymbol{0}\}$, one finds two radial locations, namely 
\begin{equation*}
r_{1, 2} = M \pm \sqrt{M^2 + l^2} =: r_l^{\pm}\, ,
\end{equation*}
at which they are not satisfied and therefore the Hamiltonians undergo ellipticity-hyperbolicity transitions. Hence, the Taub--NUT black hole spacetime has stationary limit surfaces at $\mathbb{R} \times \{r_l^{\pm}\} \times (0, \pi) \times [0, 2 \pi)$. Moreover, in the region $(- \infty, r_l^-) \cup (r_l^+, \infty) \times (0, \pi) \times [0, 2 \pi)$ the Hamiltonians are elliptic, whereas in the region $(r_l^-, r_l^+) \times (0, \pi) \times [0, 2 \pi)$ they are hyperbolic with respect to any $\boldsymbol{\eta} \in T^{\star}_{\boldsymbol{y}} \mathfrak{N} \backslash \{\boldsymbol{0}\}$ satisfying Inequality (\ref{TNI1}) and Inequality (\ref{TNI2}) for all $\boldsymbol{\xi} \in T^{\star}_{\boldsymbol{y}} \mathfrak{N} \backslash \{\boldsymbol{0}\}$.
\end{Example}

\section{Event Horizons and Black Hole Entropy} \label{SectionIV}

\noindent Since the stationary limit surfaces of black hole spacetimes with static regions coincide with their event horizons, our functional analytic detection method can be used to locate (slices of) this particular type of event horizon as shown in detail in Example \ref{SdS} in the preceding section for the black hole and cosmological event horizons of the Schwarzschild--de Sitter black hole spacetime. Moreover, it gives rise to the following functional analytic definition of an event horizon in a black hole spacetime with static regions.
\begin{Def}[Event horizon]
An event horizon $\mathfrak{H}$ of a black hole spacetime with static regions is any hypersurface on which the Dirac, Klein--Gordon, Maxwell, and Fierz--Pauli Hamiltonians undergo ellipticity-hyperbolicity transitions. For the case of an asymptotically flat black hole spacetime with static regions, an event horizon can be explicitly defined as any hypersurface of the form\footnote{Here, we consider only ordinary black holes, which do not feature more than two event horizons. Other, more exotic (multi-horizon) black holes with more than two event horizons, such as the ones studied in \cite{NojiriOdintsov}, are not taken into account.}
\begin{equation*}
\mathfrak{H} = \mathbb{R} \times \bigl\{\textnormal{ext}\bigl(\textnormal{det}(\sigma_{\mathcal{H}})^{-1}(0)\bigr)\bigr\} \times (0, \pi) \times [0, 2 \pi) \, ,
\end{equation*}
where $\textnormal{det}(\sigma_{\mathcal{H}})^{-1}(0)$ denotes the zero sets of the determinants of the principal symbols of the Hamiltonians and $\textnormal{ext}(\, . \,) \in \bigl\{\textnormal{min}(\, . \,), \textnormal{max}(\, . \,)\bigr\}$. 
\end{Def}

Based on this relational notion of an event horizon, we may propose an alternative concept of black hole entropy for asymptotically flat black hole spacetimes with static regions.
\begin{Prp}[Entropy]
The entropy $S_{\textnormal{BH}}$ of an asymptotically flat black hole spacetime with static regions is proportional to the squares of the largest zeros of the determinantial polynomials of the principal symbols of the Dirac, Klein--Gordon, Maxwell, and Fierz--Pauli Hamiltonians, i.e., 
\begin{equation} \label{BHentropy}
S_{\textnormal{BH}} = \mathscr{C}_0 \, \textnormal{max}\bigl(\textnormal{det}(\sigma_{\mathcal{H}})^{-1}(0)\bigr)^2 \, ,
\end{equation}
where $\mathscr{C}_0 \in \mathbb{R}_{> 0}$ is a constant. 
\end{Prp}
\begin{proof}
For every asymptotically flat black hole spacetime with static regions, the largest zeros of the determinantial polynomials of the principal symbols of the Hamiltonians occur at the location of the outer event horizon, that is, $\textnormal{max}\bigl(\textnormal{det}(\sigma_{\mathcal{H}})^{-1}(0)\bigr) = r^{+}_{\mathfrak{H}}$. Setting the value of the proportionality constant to $\mathscr{C}_0 = \pi$, the entropy specified in Equation (\ref{BHentropy}) is in accordance with the usual general relativistic black hole entropy expression $S_{\textnormal{BH}} = A/4$, where $A = 4 \pi (r^{+}_{\mathfrak{H}})^2$ is the outer event horizon area. 	
\end{proof}
\noindent
This particular concept of black hole entropy conveys that the farther away the outermost locus where the Hamiltonians go through an ellipticity-hyperbolicity transition is situated from a black hole's center, which relates to a more extensive sphere of action of the black hole, the larger is its entropy. 

\vspace{0.15cm}

\section*{Acknowledgments}
\noindent The author is grateful to Katharina Proksch and the Lichtenberg Group for History and Philosophy of Physics of the University of Bonn for useful discussions and comments on the first draft of this paper. This work was partially supported by the DFG Research Group FOR 2495 and via the research project MTM2016-78807-C2-1-P funded by MINECO and ERDF. 

\vspace{0.15cm}


\end{document}